  \documentclass[9pt]{article}
\usepackage{spconf}
\usepackage{IEEEtrantools}

\usepackage{amsmath,amsthm,relsize}
\usepackage{amsfonts, amssymb, cuted}
\usepackage{cleveref}
\usepackage{mathtools}
\usepackage[boxed]{algorithm}
\usepackage[noend]{algpseudocode}
\usepackage{cite}
\usepackage{xcolor}
\usepackage{soul}
\usepackage{graphicx}
\usepackage{tikz}
\usepackage{pgfplotstable}
\usepackage{pgfplots}
\usepackage{nicefrac}
\usepackage{enumitem}
\usepackage{support-caption}
\usepackage{subcaption}

\usepackage[font=footnotesize]{caption}
\usepackage{multirow}
\pgfplotsset{compat=1.15}
\usepackage{relsize}
\usetikzlibrary{patterns}
\usepackage{acro}
\usepackage{pifont}

\newcommand*{\Resize}[2]{\resizebox{#1}{!}{$#2$}}%

\newtheorem{lem}{Lemma}

\theoremstyle{definition}

\newtheorem{remarknum}{Remark} 

\newcommand{\CF}[0]{{\mathcal{F}}}

\newcommand{\CK}[0]{{\mathcal{K}}}

\newcommand{\CP}[0]{{\mathcal{P}}}

\newcommand{\CT}[0]{{\mathcal{T}}}

\newcommand{\Bx}[0]{{\mathbf{x}}}
\newcommand{\By}[0]{{\mathbf{y}}}
\newcommand{\Bz}[0]{{\mathbf{z}}}

\newcommand{\BH}[0]{{\mathbf{H}}}
\newcommand{\BI}[0]{{\mathbf{I}}}

\newcommand{\BK}[0]{{\mathbf{K}}}

\newcommand{\BQ}[0]{{\mathbf{Q}}}

\usepackage{romannum}

\title{
Multicast Transmission Design with Enhanced DoF\\ for MIMO Coded Caching Systems
}

\name{Mohammad NaseriTehrani, MohammadJavad Salehi and Antti T\"olli
}
\address{Centre for Wireless Communications, University of Oulu, Finland \\
\textrm{E-mail: \{mohammad.naseritehrani, mohammadjavad.salehi, antti.tolli\}}@oulu.fi}

\begin{document}


\maketitle

\begin{abstract}
Integrating coded caching (CC) 
into multi-input multi-output (MIMO) setups significantly enhances the achievable degrees of freedom (DoF). 
We consider a cache-aided MIMO configuration with a CC gain $t$, where a server with $L$ Tx-antennas communicates with $K$ users, each equipped with $G$ Rx-antennas.
Similar to existing works, we also extend a core CC approach, designed initially for multi-input single-output (MISO) scenarios, to the MIMO setup. However, in the proposed MIMO strategy, rather than replicating the transmit scheme from the MISO setup, the number of users $\Omega$ served in each transmission is fine-tuned to maximize DoF. 
As a result, an optimized DoF of ${\max_{\beta, \Omega }}{\Omega \beta}$ is achieved, where ${\beta \le \mathrm{min}\big(G,\nicefrac{L \binom{\Omega-1}{t}}{1 + (\Omega - t-1)\binom{\Omega-1}{t}}\big)}$ is the number of parallel streams decoded by each user. 
For the considered MIMO-CC setup, we also introduce an effective multicast transmit covariance matrix design for the symmetric rate maximization objective solved iteratively via successive convex approximation (SCA). Finally, numerical simulations verify the enhanced DoF and improved performance of the proposed design.
\end{abstract}


\section{Introduction}
\label{section:intro}


\color{black} 
\par
The expanding demand for multimedia content, mobile immersive viewing, and extended reality applications
are driving continuous growth in mobile data traffic. 
This has spurred the development of novel techniques, such as coded caching (CC)~\cite{maddah2014fundamental}, which stands out for its intriguing potential of offering a performance boost $t$ proportional to the cumulative cache size of all network users. CC ingeniously leverages network devices' onboard memory as a communication resource, especially beneficial for cacheable multimedia content. The performance gain in CC arises from multicasting well-constructed codewords to user groups of size $t+1$, allowing each user to eliminate unwanted message parts using its cached content. 
While CC was originally designed for single-input single-output (SISO) setups~\cite{maddah2014fundamental}, later studies showed it could also be used in multiple-input single-output (MISO) systems, demonstrating that spatial multiplexing and coded caching gains are additive~\cite{shariatpanahi2018physical}. 
This is achieved by serving 
multiple groups of users simultaneously with multiple multicast messages and suppressing the intra-group interference by beamforming.
Accordingly, in a MISO-CC setting with $L$ Tx antennas, $t+L$ users can be served in parallel, and the so-called degree-of-freedom (DoF) of $t+L$ is achievable~\cite{shariatpanahi2016multi,shariatpanahi2018physical}. In other works, authors in~\cite{tolli2017multi,tolli2018multicast} discussed how multi-group multicast optimized beamformers could improve the performance of MISO-CC schemes in the finite signal-to-noise ratio (SNR) regime. In the same works, the spatial multiplexing gain and the number of partially overlapping multicast messages were flexibly adjusted for a trade-off between design complexity and finite-SNR performance. 

While MISO-CC has been extensively studied in the literature, multi-input multiple-output (MIMO)-CC has received less attention. In~\cite{cao2017fundamental}, the optimal DoF of cache-aided MIMO networks with three transmitters and three receivers were studied, and in~\cite{cao2019treating}, general message sets were used to introduce inner and outer bounds on the achievable DoF of MIMO-CC schemes. More recently, low-complexity MIMO-CC schemes for single-transmitter setups were studied in~\cite{salehi2021MIMO}, and it was shown that with $G$ antennas at each receiver, if $\frac{L}{G}$ is an integer, the single-shot DoF of $Gt+L$ is achievable. Moreover, unicast and multicast beamforming strategies for improving finite-SNR MIMO-CC performance were explored in~\cite{salehi2023multicast}.

In this paper, we extend the works of~\cite{salehi2023multicast,salehi2021MIMO} on MIMO-CC schemes. The contributions are two-fold. First, we alter the delivery scheme of~\cite{salehi2021MIMO} to improve the achievable DoF. The idea is that, instead of serving the same number of users as the baseline MISO-CC scheme in each transmission, we judiciously select this parameter to maximize the single-shot DoF. As a result, the integer constraint of~\cite{salehi2021MIMO} on $\frac{L}{G}$ is removed and the improved DoF of $\max_{\beta, \Omega }{\Omega \beta}$, where ${\beta \le \mathrm{min}\big(G,\nicefrac{L \binom{\Omega-1}{t}}{1 + (\Omega - t-1)\binom{\Omega-1}{t}}\big)}$, is achieved. Second,  
we extend the linear multicast beamformer design of~\cite{salehi2023multicast}, which had limited flexibility in the design parameters, by formulating a symmetric rate maximization problem with respect to (w.r.t) the general transmission covariance matrices of multicast messages for a given delivery scheme, and where the user-specific rate is aggregated from all multicast messages intended to the given user. As the resulting problem is non-convex, we use successive convex approximation (SCA) to find the symmetric rate, which exceeds the rate of~\cite{salehi2023multicast}. Numerical simulations confirm the enhanced DoF and improved performance. 

Throughout the text, we use these notations: For any integer $\Resize{0.2cm}{J}$, $\Resize{2.2cm}{[J] \equiv \{1,2,\dots,J\}}$. Bold upper- and lower-case letters signify matrices and vectors, respectively. Calligraphic letters denote sets, $\Resize{.45cm}{|\CK|}$ denotes set size, and $\Resize{.65cm}{\CK \backslash \CT}$ represents elements in $\CK$ excluding those in $\Resize{.25cm}{\CT}$. $\mathsf B$ is the collection of sets, with its count as $|\mathsf B|$. Additional notations are introduced as needed. \vspace{-3mm}

\section{System Model}
\label{section:sys_model}
\subsection{Network Setup}
\label{section:network setup}
We consider a MIMO setup where a single BS with $L$ transmit antennas serves $K$ cache-enabled users each with $G$ received antennas, as shown in Figure~\ref{fig:ISIT_sysm}.\footnote{In general,
$L$ and $G$ denote spatial multiplexing capability, potentially smaller than the actual number of antennas due to rank and RF limitations.} 
Every user has a cache memory of size $M$ data units and requests files from a library $\CF$ of $N$ unit-sized files.
The coded caching gain is defined as $t \equiv \frac{KM}{N}$, which represents how many copies of the file library could be stored in the cache memories of all users. The system operation comprises two phases; placement and delivery. In the placement phase, the users' cache memories are filled with data. Following a similar structure as~\cite{shariatpanahi2016multi}, we split each file $W \in \CF$ into $\binom{K}{t}$ subfiles $W_{\CP}$, where $\CP \subseteq [K]$ denotes any subset of users with $|\CP| = t$. Then, we store  $W_{\CP}$, $\forall W \in \CF,\forall \CP : k \in \CP$ in the cache memory of user $k \in [K]$. 
\begin{figure}[t]
        \centering
        \includegraphics[height = 4cm]{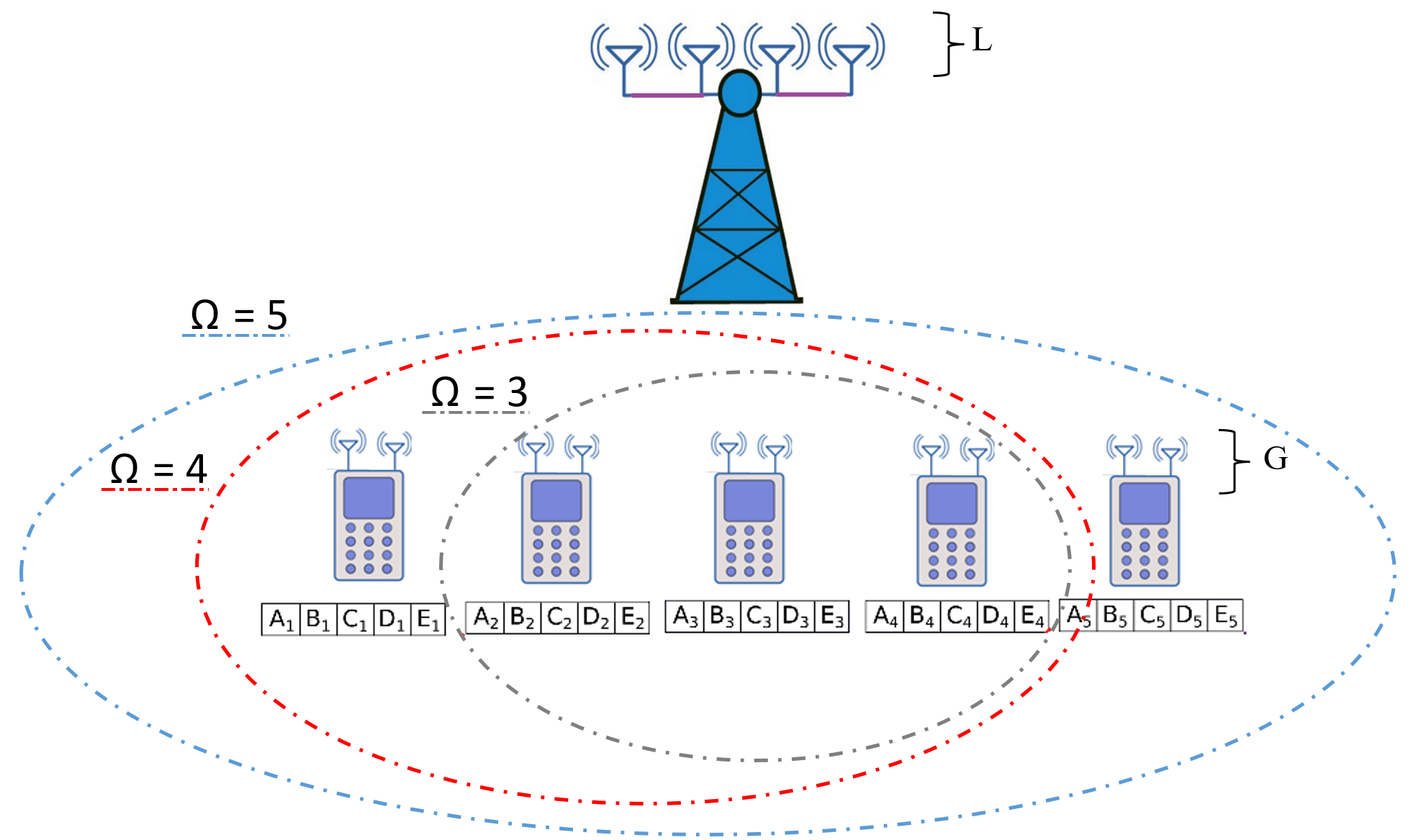} \label{fig:receiver_Prop}
    \caption{MIMO CC system model and the selection of users for different $\Omega$.}\label{fig:ISIT_sysm}
    \vspace{-5mm}
\end{figure}

At the beginning of the delivery phase, each user $k$ reveals its requested file $W(k) \in \CF$ to the server. The server then constructs and transmits (e.g., in consecutive time slots) a set of transmission vectors $\Bx(i) \in \mathbb{C}^L$,
$i \in [\binom{K}{\Omega}]$, where $\Bx(i)$ delivers parts of the requested data to every user in a subset $\CK(i)$ of users with $|\CK(i)| = \Omega$, and $t+1 \le \Omega \le t+L$ is the number of users served in each transmission, chosen to maximize the achievable DoF (detailed explanation in Section~\ref{section:Achievable DoF}). 
Upon transmission of $\Bx(i)$, user $k \in \CK(i)$ receives
$\By_k(i) = \BH_k \Bx(i) + \Bz_k(i)$,
where $\BH_k \in \mathbb{C}^{G \times L}$ is the i.i.d channel matrix with normal distribution between the server and user $k$ which is available at the server, and $\Bz_k(i) \sim \mathcal{CN}(\mathbf{0},N_0 \mathbf{I})$ represents the noise. 
In order to define the symmetric rate, we need to know the length (in data units) of each transmission vector. Using a similar reasoning as~\cite{shariatpanahi2018physical}, to ensure new data is delivered in each transmission, we need to further split every subfile $W_{\CP}$ into $\binom{K-t-1}{\Omega-t-1}$ equal-sized subpackets $W_{\CP}^q$. Moreover, as will be seen shortly (Section~\ref{section:trans_vector_building}), the codeword creation process in the delivery phase involves XOR operations over the finite field that do not alter the data size. As a result, the length of every transmission vector is the same as the subpacket length of $\frac{1}{\Theta}$, where $\Theta = \binom{K}{t}\binom{K-t-1}{\Omega-t-1}$ denotes the \emph{subpacketization} level. Now, using $R_i$ to denote the multicast (max-min) transmission rate of $\Bx(i)$ that enables successful decoding at every user $k \in \CK(i)$, the transmission time of $\Bx(i)$ is $T_i = \frac{1}{\Theta R_i}$, and the symmetric rate can be defined as\vspace{-3mm}
\begin{equation}
\label{eq:symmetric_rate} \vspace{-2mm}
    R_{sym} = \frac{K}{\sum_i T_i}= \frac{K}{\sum_i {\frac{1}{\Theta R_i}}} =  \Resize{2.3cm}{\frac{K {\binom{K}{t}\binom{K-t-1}{\Omega-t-1}}}{{\sum_i{\frac{1}{R_i}}}}}  \; .
\end{equation} \vspace{-4mm}
The goal is to design the delivery scheme to maximize $R_{sym}$. \vspace{-3mm}

\subsection{Building the Transmission Vectors $\Bx(i)$}
\label{section:trans_vector_building}
The proposed cache placement strategy, which is a straightforward adaptation of~\cite{shariatpanahi2016multi}, enables us to create an XOR codeword for every subset $\CT$ of users with $|\CT| = t+1$, such that the data intended for every user is available in the cache memory of all the other $t$ users of $\CT$. 
In our general MIMO-CC setup, we build a separate transmission vector $\Bx(i)$ for every subset $\CK(i)$ of users with $|\CK(i)| = \Omega$. Following the same intuition of SISO~\cite{maddah2014fundamental} and MISO~\cite{shariatpanahi2018physical} systems, to build $\Bx(i)$, we also create a multicast message for every subset $\CT$ of $\CK(i)$ with $|\CT| = t+1$. However, to mitigate the interference among these multicast messages, instead of the stream-specific beamforming design~\cite{tolli2017multi,salehi2023multicast,mahmoodi2021low}, we consider a general multicast transmission design as\vspace{-3mm}
\begin{equation}\vspace{-2mm}
\label{eq:TX_signal}
\Bx(i) = \Resize{0.4cm}{\sum}_{\mathcal{T}\subseteq \CK(i), |\mathcal{T}|= t+1 }{\Bx}_{\mathcal{T}}(i) ,
\end{equation}
where ${\Bx}_{\mathcal{T}}(i)$ is the corresponding multicast signal for users in $\CT$ at time interval $i$, chosen from a 
complex Gaussian codebook ${\Bx}_{\mathcal{T}} \!\sim\! \Resize{1.6cm}{\mathcal{CN}(\mathbf{0}, \BK_{{\Bx}_{\mathcal{T}}})}$ with a 
covariance matrix $\Resize{.6cm}{\BK_{{\Bx}_{\mathcal{T}}}}$.
Without loss of generality, let us consider a single time interval and remove the $i$ index (the same process is repeated at every interval). Let us also define $\Resize{3.6cm}{{\mathsf S}^{\CK} = \{\CT \subseteq \CK, |\CT| =  t +1 \}}$ as the set of all multicast groups in the considered time interval, and $\Resize{3.1cm}{{\mathsf S}_k^{\CK} = \{\CT \in {\mathsf  S}^{\CK}\! \mid \! k \in \CT \}}$ and $\Resize{1.7cm}{\bar{\mathsf S}_k^{\CK} =  {\mathsf S}^{\CK} \backslash {\mathsf S}_k^{\CK}}$ as the set of multicast groups including and not including user $k$, respectively. Then, we can rewrite the transmission vector in~\eqref{eq:TX_signal} as $\Bx = \sum_{\CT \in {\mathsf S}^{\CK}}{\Bx}_{\mathcal{T}}$, and similarly, 
the received signal at user $k \in \CK$) can also be reformulated as\vspace{-2mm}
 \begin{equation}\vspace{-2mm}
\label{eq:RX_signal_MAC}
\By_k = \Resize{.98cm}{\BH_k \sum}_{\CT\in {\mathsf S}_k^{\CK}}\:{\Bx}_{\CT} + \Resize{.98cm}{\BH_k \sum}_{\CT\in \bar{\mathsf S}_k^{\CK} }\:{\Bx}_{{\mathcal{T}}} + \Bz_k,\vspace{-1.mm}
\end{equation}
where the first and the second summation represent the intended and interference terms for user $k$, respectively. If $\Omega > t+1$, every user $k \in \CK$ should decode multiple codewords received in parallel and the codeword-specific rates are constrained by the equivalent user-specific multiple access channel (MAC) rate region with a size of $m_k = \binom{\Omega-1}{t}$ ~\cite{tolli2017multi,shariatpanahi2018physical}. More details are provided in Section~\ref{section:Multi-Group Multicast Communication in Coded caching}.




\section{Improved DoF for MIMO-CC Systems}
\label{section:Achievable DoF}
CC techniques provide the total DoF of $t+1$ in SISO setups~\cite{maddah2014fundamental}. This DoF is increased to $t+L$ in a MISO setup with the spatial multiplexing gain of $L$ at the transmitter~\cite{shariatpanahi2018physical}, and is shown to be optimal with uncoded cache placement and single-shot data delivery~\cite{lampiris2021resolving}. Denoting the DoF as the number of parallel data streams (and not the number of users),
in~\cite{salehi2021MIMO}, it is shown that if $\frac{L}{G}$ is an integer, the total DoF of a MIMO-CC setup can reach $Gt+L$, where $G$ is the spatial multiplexing gain at the receiver side. 

Consider a MIMO-CC network with the number of users $K \le t+L$.
Assume in every transmission, we serve a total number of $\Omega$ users, and each user wants to decode $\beta$ streams simultaneously. The goal is to maximize $\Omega \beta$ while assuring that decoding all the parallel streams is possible and not limited by interference. This requires that $\beta$ is limited by the rank of the equivalent MIMO-CC channel for each user. The following Lemma clarifies how $\beta$, $\Omega$, and the rank of the equivalent channel are interrelated.  

\begin{lem}
\label{lem:channel_rank}
For the considered MIMO-CC scheme we have:\\
${\mathrm{Rank}(\BH_{k,\mathrm{MIMO-CC}})\! \le\! \min \!\big(G, (L - (\Omega - t - 1)\beta) \binom{\Omega - 1}{t}\!\big).}$
\end{lem}
\begin{proof}
Consider an arbitrary user $k$ receiving data from ${\binom{\Omega - 1}{t}}$ multicast messages sent during an arbitrary interval. With CC gain $t$, the interference caused by every stream sent to user $k$ can be removed by cache contents at up to $t$ users. So, for successful decoding, the interference caused by streams sent to user $k$ should be suppressed by transmitter-side precoding over every stream intended for the rest of ${\Omega-t-1}$ target users; i.e., with the Tx-side spatial multiplexing of $L$, after the interference cancellation for each multicast message, the remaining spatial DoF will be ${L-(\Omega-t-1)\beta}$. However, we have ${\binom{\Omega - 1}{t}}$ total multicast messages, and each of them requires interference cancellation over a different set of streams, which means the remaining spatial multiplexing gain is ${\big(L - (\Omega - t - 1)\beta) \binom{\Omega - 1}{t}}$. On the other hand, the number of decoded streams is upper-bounded by $G$ by definition, and the proof is complete.
%
\end{proof}

Using Lemma~\ref{lem:channel_rank}, the maximum achievable DoF for the proposed MIMO-CC transmission design is given by solving
%
\begin{IEEEeqnarray}{lll}\vspace{-1mm}
\label{eq:total_DoF}
 & {\mathrm{DoF}_{\max}} =  \max_{\beta , \Omega }~ \Omega \beta,   \nonumber\\\vspace{-2.mm}
    &s.t.\:\:{\beta \le \mathrm{min}\big(\Resize{0.25cm}{G},\frac{L \binom{\Omega-1}{t}}{1 + (\Omega - t-1)\binom{\Omega-1}{t}}\big).}\vspace{-1mm}
\end{IEEEeqnarray}
The function inside the $\max$ operation is \emph{not} monotonic w.r.t $\beta$, and so the maximum DoF may not occur at $\beta = G$ (e.g., for $L=16$, $G=4$, $t=1$, the DoF of~$21$ is achieved with $\beta = 3$, while with
$\beta = G = 4$, the DoF is~$20$). 
Nevertheless, setting $\beta = G$ gives us a quick metric to calculate an achievable DoF value ${\mathrm{DoF}_G = G \lfloor \frac{L-1}{G} \rfloor+ G(t+1)}$. For example, if $L=3$, $G=\beta=2$, and $t=1$, $\mathrm{DoF}_{G}=6$ is achievable. Note that with the MIMO-CC scheme in~\cite{salehi2021MIMO}, we need at least $L=4$ to achieve the same DoF. So, our scheme not only 
relaxes the integer constraint on $\frac{L}{G}$
but also shows that the gain boost of MIMO-CC setups is possible with smaller transmitter-side spatial multiplexing gains than suggested in~\cite{salehi2021MIMO}. 
\vspace{-3mm}

\section{Multigroup Multicast Transmission}
\label{section:Multi-Group Multicast Communication in Coded caching}
\vspace{-2mm}In this work, we examine the worst-case delivery rate at which the system can serve all requests for files from the library. 
We propose an extension of linear MISO multigroup multicasting~\cite[Theorem 1]{tolli2018multicast} to a general transmission strategy in MIMO-CC systems. The goal is to optimize the minimum user-specific rate over the MIMO capacity region by developing a generalized multicast transmission strategy. The user-specific rate at a particular user $k$ is obtained by aggregating all the messages that form an equivalent MAC channel, i.e., $R_{MAC}^k  = \min_{\mathsf B \subseteq {\mathsf S}_k^{\CK}}[\frac{1}{|\mathsf B|} \sum_{{\CT}\in \mathsf B} R_{\CT} ]$, where $\mathsf B \neq \emptyset$ could represent any set family of codewords with sizes $|\mathsf B| = i,i \in \{1,\dots, {m_k} \}$, and $\sum_{{\CT}\in \mathsf B} R_{\CT}$ limits the individual and sum rate of any combination of two or more transmitted multicast signals determined with sizes of $|\mathsf B| $~\cite{tolli2018multicast}. To this end, we formulate the symmetric rate maximization problem for the MIMO multigroup multicasting model in~\eqref{eq:RX_signal_MAC} as follows: \vspace{-2mm}
\begin{IEEEeqnarray}{lll}
\label{eq:P_optimum1}
{\max_{\BK_{{\Bx}_{\CT}}, R_{\CT}, \CT\in \mathsf B}} ~~\min_{k\in \CK} 
\underset{{\mathsf B} \subseteq {\mathsf S}_k^{\CK}}{\text{min}}\Big[{\frac{1}{|{\mathsf B}|} \sum_{{\CT}\in {\mathsf B}} R_{\CT}} \Big]
  \nonumber\\
    s.t. \quad {\sum_{{\CT}\in \mathsf B }R_{\CT} \leq 
    \log| \BI+\BH_k \sum_{{\CT}\in \mathsf B }\BK_{{\Bx}_{\CT}}\BH_k^H\BQ_{k}^{-1}|}, \; \nonumber\\
  {\forall k \in \CK, \mathsf B \subseteq {\mathsf S}_k^{\CK} 
     \quad \quad \sum_{{\CT}\in {\mathsf S}^{\CK} }\text{Tr}(\BK_{{\Bx}_{\CT}})
	  \leq P_T},\vspace{-2mm}
\end{IEEEeqnarray}
where the first constraint covers the rate region among all multicast messages received by user $k$, $P_T$ is the power budget and ${\BQ_{k}  = (N_0 \BI + \BH_k \sum_{\CT\in \bar{\mathsf S}_k^{\CK} }\BK_{{\Bx}_{{\mathcal{T}}}} \BH_k^H)}$. 
The optimization problem \eqref{eq:P_optimum1} is NP-hard and non-convex due to the interference terms ${\BQ_{k}}$. Utilizing SCA for each constraint $ \Resize{1.9cm}{k \in \CK, \mathsf B \subseteq \mathsf S_k^{\CK}}$, enables an iterative solution. These constraints are expressed as the difference of two convex functions (\eqref{eq:Taylor_app1}), allowing the application of first-order Taylor series approximation to lower bound the negative logdet terms (\eqref{eq:Taylor_app2}):
\vspace{-2.0mm}
\begin{IEEEeqnarray}{lll}
\label{eq:Taylor_app1}\small
  {\sum_{{\CT}\in \mathsf B }R_{\CT} \leq \log\Big| \BQ_k+\BH_k  \sum_{{\CT}\in \mathsf B }\BK_{{\Bx}_{\CT}} \BH_k^H\Big|- \log|\BQ_k|}, \;\end{IEEEeqnarray}\vspace{-8.mm}\begin{IEEEeqnarray}{lll} \label{eq:Taylor_app2}\small 
  {\mathrm{Tr}\big(\bar{\BQ}_{k}^{-1}  \BH_k \sum_{\CT\in \bar{\mathsf S}_k^{\CK} }(\bar{\BK}_{{\Bx}_{{\mathcal{T}}}}-\BK_{{\Bx}_{{\mathcal{T}}}})\BH_k^H\big) -\log|\bar{\BQ}_{k}|}.\vspace{-2.3mm} 
\end{IEEEeqnarray}
Here the constants ${\bar{\BQ}_{k}}$ are a function of the previous SCA solutions denoted by ${\bar{\BK}_{{\Bx}_{{\mathcal{T}}}}}$. Lastly, we can express the approximated convex optimization problem as follows:
\vspace{-2.5mm}
\begin{IEEEeqnarray}{lll}
\scriptsize
\label{eq:P_optimum2}
\max_{\BK_{{\Bx}_{\CT}}, \CT \subseteq \CK, R} ~  {R} \nonumber\\[-4pt]
s.t. \:\: {R \! \leq \! \frac{1}{|\mathsf B|}\Big(\!\log\big|N_0\BI +\BH_k (\sum_{{\CT}\in \mathsf B }\BK_{{\Bx}_{\CT}} +  \!\!\sum_{\CT\in \bar{\mathsf S}_k^{\CK} }\BK_{{\Bx}_{{\mathcal{T}}}} )\BH_k^H\big|}\nonumber\\[-3pt]
\qquad\quad {+ \mathrm{Tr} ( \bar{\BQ}_{k}^{-1}  \BH_k \!\! \sum_{\CT\in \bar{\mathsf S}_k^{\CK} }(\bar{\BK}_{{\Bx}_{{\mathcal{T}}}}-\BK_{{\Bx}_{{\mathcal{T}}}})\BH_k^H ) - \log|\bar{\BQ}_{k}|\Big)} \nonumber\\[-2pt]
{\forall k \in \CK, \mathsf B \subseteq {\mathsf S}_k^{\CK},  
\quad \quad \sum_{{\CT}\in {\mathsf S}^{\CK} }\text{Tr}(\BK_{{\Bx}_{\CT}})\leq P_T}. \vspace{-2mm} 
\end{IEEEeqnarray} Starting with a randomly initial point $\bar{\BK}_{{\Bx}_{{\mathcal{T}}}}$, $\CT \in {\mathsf S}_k^{\CK}, \forall k \in \CK$, that fulfills the power constraint and iterating until convergence, the optimization problem~\eqref{eq:P_optimum2} achieves a local solution of~\eqref{eq:P_optimum1}. The stopping criterion is met when the rate improvement $R$ over the previous SCA value $\Bar{R}$ is within a specified threshold of $Er_\text{SCA}$, i.e., ${\big|R - \Bar{R}\big|}$ $\leq$ $Er_\text{SCA}$.
 
\begin{remarknum} \label{remark1}
\normalfont For the special case of $L=G$ and $\Omega = {t + 1}$, the symmetric multicast transmission signal vector in~\eqref{eq:RX_signal_MAC} becomes interference-free, and hence,~\eqref{eq:P_optimum2} is simplified to the following convex maxmin problem:\vspace{-2.5mm}    
\begin{IEEEeqnarray}{lll}
\label{eq:P_optimum1_K1_sym}
\max_{\BK_{{\Bx}_{\CT}}, \CT \subseteq \CK, R_{}^k, \forall k \in \CK} \min_{k\in \CK} 
  (R_{}^1, \dots, R_{}^k) \nonumber\\
    s.t.\quad {R_{}^{k} \leq \log|\BI+ N_0^{-1} \BH_k \BK_{{\Bx}_{\CT}}\BH_k^H|, \; \forall k \in \CK} \nonumber\\
     \quad\:\quad {\text{Tr}(\BK_{{\Bx}_{\CT}})
	  \leq P_T,}\vspace{-2mm} 
\end{IEEEeqnarray}
where $R_{}^k$ represents the rate of $k$-th user receiving a particular multicast signal.
\end{remarknum}

\vspace{-4mm}
\section{Numerical Results}
\vspace{-1mm}
\label{section:NumRes}


\begin{figure}
    \centering
    \resizebox{0.9\columnwidth}{!}{%

    \begin{tikzpicture}

    \begin{axis}
    [
    axis lines = center,
    xlabel near ticks,
    xlabel = \smaller {SNR [dB]},
    ylabel = \smaller {Symmetric Rate [bits/s]},
    ylabel near ticks,
    ymin = 0,
    xmax = 30,
    legend pos = north west,
    ticklabel style={font=\smaller},
    grid=both,
    major grid style={line width=.2pt,draw=gray!30},
    ]

   \pgfplotsset{
    myaxis/.style={
        axis lines=left,
        grid=major,             
        grid style={dotted,gray!70}, 
        legend style={
            minimum height=0.5em,
            minimum width=0.5em,
            draw=black,
            fill=none,
        },
    },
     }

    
    \addplot
    [dashed, mark = o, blue]
    table[y=RsymLG22K2t1logdet,x=SNR]{Figs/data.tex};
    \addlegendentry{\tiny \textit{L} = $2$, \textit{G} = $2$, $\Omega$ = $2$, \textit{t} = $1$, DoF = 4
    }
    
    
    \addplot
    [ mark = square, blue]
    table[y=RsymLG22K3t1,x=SNR]{Figs/data.tex};
    \addlegendentry{\tiny\textit{L} = $2$, \textit{G} = $2$, $\Omega$ = $3$, \textit{t} = $1$, DoF = 3
    }
    
    \addplot
    [dashed, mark = star, red!80]
    table[y=RsymLG22K3t2,x=SNR]{Figs/data.tex};
    \addlegendentry{\tiny \textit{L} = $2$, \textit{G} = $2$, $\Omega$ = $3$, \textit{t} = $2$, DoF = 6
    }
    
    \addplot
    [mark = o, red!80]
    table[y=RsymLG22K4t2,x=SNR]{Figs/data.tex};
    \addlegendentry{\tiny \textit{L} = $2$, \textit{G} = $2$, $\Omega$ = $4$, \textit{t} = $2$, DoF = 4
    }


     \addplot
    [dash dot, mark = o, black!90]
    table[y=RsymLG33K2t1CovMC,x=SNR]{Figs/data.tex};
    \addlegendentry{\tiny \textit{L} = $3$, \textit{G} = $3$, $\Omega$ = $2$, \textit{t} = $1$, DoF = 6 
    }

     \addplot
    [dashed, mark = x, black!90]
    table[y=RsymLG33K2t1ICASSPMC,x=SNR]{Figs/data.tex};
    \addlegendentry{\tiny \textit{L} = $3$, \textit{G} = $3$, $\Omega$ = $2$, \textit{t} = $1$, DoF = 6\cite{salehi2023multicast}}

\textit{L} =
    
    \end{axis}

    \end{tikzpicture}
    }
\vspace{-0.3cm}
    \caption{MIMO multicast design for $L = G$ }
    \label{fig:plot_2}
    \vspace{-0.3cm}
\end{figure}
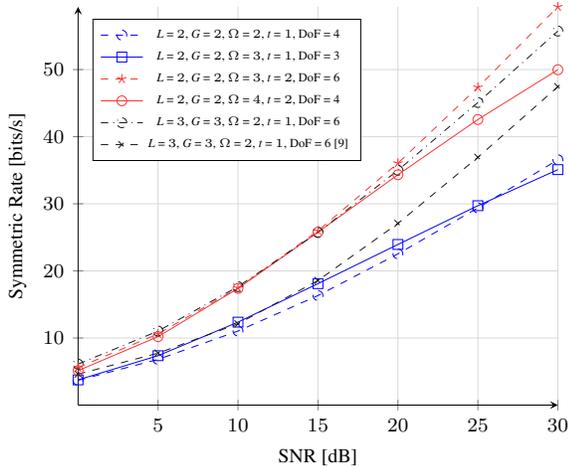
Numerical results are generated for different combinations of network and design parameters $L$ $G$, $t$, and $\Omega$ to compare various transmission strategies. 
 The SNR is defined as $\frac{P_T}{N_0}$, where $N_0$ is the fixed noise variance, and $\Resize{1.9cm}{Er_\text{SCA}\leq 10^{-4}}$. 

Fig.~\ref{fig:plot_2} compares the symmetric rate of the $L=G$ case, for different $\Omega$, $t$, and SNR. We observe that at high SNR, the slopes of the symmetric rate curves align with DoF expectations (higher DoF, steeper slope, higher symmetric rate). Interestingly, the $L=G=2$ case exhibits the same DoF value of 4, for both $\Omega=2$, $t=1$ and $\Omega=4$, $t=2$, even though the CC gain is larger in the latter case (of course, the symmetric rate is higher with $t=2$ due to increased multicasting gains). This highlights how choosing $\Omega$ affects the DoF. In fact,
choosing the proper $\Omega$ value ($\Omega=2$ for $t=1$ and $\Omega = 3$ for $t=2$) enables achieving Rank-2 transmit covariance matrices (full rank), allowing more parallel streams in total. In contrast, with a larger $\Omega$, we will be limited to rank-1 covariance matrices, and the results will be essentially similar to the MISO case~\cite{tolli2017multi} albeit with a constant positive rate shift due to additional receiver combining.

The same figure shows that for $L=G=2$, $t =1$, choosing $\Omega = 3$ yields a better symmetric rate than $\Omega = 2$ despite achieving a smaller DoF. This is because the system is fully loaded with $\Omega = 2$, constraining the beamforming gain at low SNR~\cite{tolli2017multi}. Finally, comparing our transmission strategy with~\cite{salehi2023multicast} for the $L=G=3$ case, it can be seen that our approach consistently outperforms the other, setting an upper bound on the performance of its linear transceiver design.

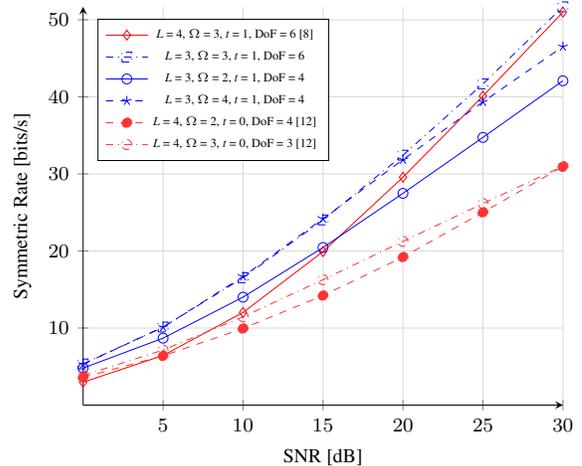
\begin{figure}
    \centering
    \resizebox{0.9\columnwidth}{!}{%

    \begin{tikzpicture}

    \begin{axis}
    [
    axis lines = center,
    xlabel near ticks,
    xlabel = \smaller {SNR [dB]},
    ylabel = \smaller {Symmetric Rate [bits/s]},
    ylabel near ticks,
    ymin = 0,
    xmax = 30,
    legend pos = north west,
    ticklabel style={font=\smaller},
    grid=both,
    major grid style={line width=.2pt,draw=gray!30},
    ]
    


    \addplot
    [mark = o, mark = diamond, red]
    table[y=RsymLG42K3t1WSAZFUnicast,x=SNR]{Figs/data.tex};
    \addlegendentry{\tiny \textit{L} = 4,  $\Omega$ = 3, \textit{t} = 1, DoF  = 6 \cite{salehi2021MIMO} }

        \addplot
    [dash dot, mark = square, blue]
    table[y=RsymLG32K3t1,x=SNR]{Figs/data.tex};
    \addlegendentry{\tiny \textit{L} = 3,  $\Omega$ = 3, \textit{t} =  1, DoF = 6}

        \addplot
    [ mark = o, blue]
    table[y=RsymLG32K2t1,x=SNR]{Figs/data.tex};
    \addlegendentry{\tiny \textit{L} = 3,   $\Omega$ = 2, \textit{t} = 1, DoF = 4}
    
    \addplot
        [dashed, mark = star, blue]
        table[y=RsymLG32K4t1,x=SNR]{Figs/data.tex};
        \addlegendentry{\tiny \textit{L} = 3,  $\Omega$ = 4, \textit{t} = 1, DoF = 4 }


    

    \addplot
        [dashed, mark = * , red!80]
        table[y=RsymLG42K2t0MIMOUnicast,x=SNR]{Figs/data.tex};
        \addlegendentry{\tiny \textit{L} = 4, $\Omega$ = 2, \textit{t} = 0, DoF  = 4 \cite{biglieri2007mimo}}

 \addplot
        [dash dot, mark = o , red!80]
        table[y=RsymLG42K3t0MIMOUnicast,x=SNR]{Figs/data.tex};
        \addlegendentry{\tiny \textit{L} = 4, $\Omega$ = 3, \textit{t} = 0,  DoF  = 3 \cite{biglieri2007mimo}}
        

    
    
    \end{axis}

    \end{tikzpicture}
    }
\vspace{-0.3cm}
    \caption{MIMO multicast design for $L > G $ - $G =2, t=\{1,0\}$}
    \label{fig:plot_3}
\vspace{-0.3cm}
\end{figure}


In Fig.~\ref{fig:plot_3}, we have considered the case $L > G$.
Accordingly, the slope of the delivery scheme with $L = 3$, $G = 2$, and $\Omega = 3$ confirms the achievability of the DoF of 6, even though the Tx-side multiplexing gain is smaller than $L=4$ suggested by~\cite{salehi2021MIMO} (of course, the symmetric rate of the $L=4$ case still exceeds the $L=3$ case due to the larger Tx-side beamforming gain, but the DoF for both cases is 6).  
Also, for the case of $L=3$, $G=2$, by comparing $\Omega = 2$ and $\Omega = 4$, we see that despite having the same DoF, the latter setup has a significantly larger symmetric rate, which can be explained by the extra multi-user diversity gain (larger number of UEs and Rx antennas) offered. 
Lastly, to assess the impact of CC in MIMO setups, we have also included simulation results for the scenario without CC ($t=0$, baseline MIMO communication as discussed in~\cite{biglieri2007mimo}).
The outcomes validate the substantial potential of CC in enhancing the performance of MIMO systems, in terms of both DoF and the symmetric rate.

\vspace{-3.0mm}
\section{Conclusion}
\vspace{-1mm}

In this paper, we proposed a flexible CC scheme for MIMO setups by optimizing the number of users served in each transmission to maximize the achievable DoF. We also proposed a high-performance beamformer design for MIMO-CC setups by formulating the problem of maximizing the symmetric rate w.r.t transmit covariance matrices for the multicast signals. We then solved the non-convex problem using SCA and verified the performance enhancement through simulations.

\bibliographystyle{IEEEtran}
\bibliography{references,whitepaper}


\end{document}